\documentclass[11pt]{article}
\usepackage{graphicx}
\usepackage{color}
\usepackage{amsmath}
\usepackage{amssymb}
\usepackage{amscd}
\newcommand{\R}{\mathbb{R}}
\newcommand{\inr}[1]{\bigl< #1 \bigr>}

\newcommand{\E}{\mathbb{E}}

\newcommand{\eps}{\varepsilon}

\newcommand{\inner}[2]{{\langle#1,#2\rangle}}

\newcommand{\RR}{{\mathbb{R}}}

\newcommand{\hx}{{\hat{x}}}

\newcommand{\bl}{\left(}
\newcommand{\br}{\right)}

\newtheorem{Theorem}{Theorem}[section]
\newtheorem{Lemma}[Theorem]{Lemma}

\newtheorem{Definition}[Theorem]{Definition}

\newtheorem{Proposition}[Theorem]{Proposition}
\newtheorem{Corollary}[Theorem]{Corollary}

\newtheorem{Assumption}{Assumption}[section]
\newtheorem{Question}[Theorem]{Question}
\numberwithin{equation}{section}

\def \proof {\noindent {\bf Proof.}\ \ }

\def \endproof
{{\mbox{}\nolinebreak\hfill\rule{2mm}{2mm}\par\medbreak}}

\title{Phase Retrieval: Stability and Recovery Guarantees}

\begin{document}
\author{Yonina C. Eldar\thanks{
Department of Electrical Engineering,
Technion---Israel Institute of Technology, Haifa 32000, Israel. Email:yonina@ee.technion.ac.il.} \and Shahar Mendelson\thanks{Department of Mathematics, Technion---Israel Institute of Technology, Haifa 32000, Israel. Email:
shahar@tx.technion.ac.il.
\newline
The work of Y. Eldar is supported in part by the Israel Science Foundation under Grant no.
170/10, in part by the Ollendorf Foundation, and in part by a Magnet grant Metro450 from
the Israel Ministry of Industry and Trade. The work of S. Mendelson is supported in part by the Centre for Mathematics and its Applications, The Australian National University, Canberra, ACT 0200, Australia. Additional support was given by an Australian Research Council Discovery grant DP0986563, the European Community's Seventh Framework Programme (FP7/2007-2013) under ERC grant agreement 203134, and by the Israel Science Foundation grant 900/10.}
}

\maketitle

\begin{abstract}
We consider stability and uniqueness in real phase retrieval problems over general input sets.
Specifically, we assume the data consists of noisy quadratic measurements of an unknown input $x \in \R^n$ that lies in a general set $T$ and study conditions under which $x$ can be stably recovered from the measurements. In the noise-free setting we derive a general expression on the number of measurements needed to ensure that a unique solution can be found in a stable way, that depends on the set $T$ through a natural complexity parameter. This parameter can be computed explicitly for many sets $T$ of interest. For example, for $k$-sparse inputs we show that $O(k\log(n/k))$ measurements are needed, and when $x$ can be any vector in $\R^n$, $O(n)$ measurements suffice.
In the noisy case, we show that
if one can find a value for which the empirical risk is bounded by a given, computable constant (that depends on the set $T$), then the error with respect to the true input is bounded above by an another, closely related complexity parameter of the set. By choosing an appropriate number $N$ of measurements, this bound can be made arbitrarily small, and it decays at a rate faster than $N^{-1/2+\delta}$ for any $\delta>0$.
In particular,
for $k$-sparse vectors stable recovery is possible from $O(k\log(n/k)\log k)$ noisy measurements, and when $x$ can be any vector in $\R^n$, $O(n \log n)$ noisy measurements suffice.
We also show that  the complexity parameter for the quadratic problem is
  the same as the one used for analyzing stability in linear measurements under very general conditions. Thus, no substantial price has to be paid in terms of stability if there is no knowledge of the phase.
\end{abstract}

\section{Introduction}
Recently there has been growing interest in recovering an input vector $x \in \RR^n$ from quadratic measurements
\begin{equation}
\label{eq:meas}
y_i=|\inner{a_i}{x}|^2+w_i,\quad i=1,\ldots,N
\end{equation}
where $w_i$ is noise, and $a_i$ are a set of known vectors.  Since only the magnitude of $\inner{a_i}{x}$ is measured, and not the phase (or the sign, in the real case), this problem is referred to as {\em phase retrieval}.
Phase retrieval problems arise in many areas of optics, where the detector can only measure the magnitude of the received optical
wave. Several important applications of phase retrieval include X-ray crystallography, transmission electron microscopy and coherent diffractive imaging \cite{Q10,H01p,H93,W63}.

Many methods have been developed for phase recovery \cite{H01p} which often rely on  prior information about the signal,
such as positivity or support constraints. One of the most popular techniques is based on alternating projections,
where the current signal estimate is transformed back and forth between the object and the Fourier domains.
The prior information and observations are used in each domain in order to form the next estimate.
 Two of the main approaches of this type are Gerchberg-Saxton \cite{GS72} and Fienup \cite{F82}.
  In general, these methods are not guaranteed to converge, and often require careful parameter selection and sufficient prior information.

To circumvent the difficulties associated with alternating projections, more recently, phase retrieval problems have been treated using semidefinite relaxation, and low-rank matrix recovery ideas \cite{CESV12,SESE11}.
In \cite{CESV12} several masks where used in the measurement process in order to ensure the ability to retrieve the phase.
Another approach to generate robust solutions is to assume that the input signal $x$ is sparse, namely, that it contains only a few non-zeros values in an appropriate basis expansion.
Sparsity has long been exploited in signal processing, applied mathematics, statistics and computer science for tasks such as compression, denoising, model selection, image processing and more. Despite the great interest in exploiting sparsity in various applications, most of the work to date has focused on recovering sparse or low rank data from linear measurements \cite{EK12,D06,CRT06a}.
Recently, the basic sparse recovery problem has been generalized to the case in which the measurements are quadratic \cite{SESE11}, or given by a more general nonlinear transform of the unknown input \cite{BE12}.
The first paper to consider sparse phase retrieval was \cite{SESE11}, based on semidefinite relaxation combined with a row-sparsity constraint on the resulting matrix. An iterative thresholding algorithm was then proposed that approximates the solution. Similar approaches were later used in  \cite{JOH12,OYDS12}.  An alternative algorithm was recently designed in \cite{BE12,SBE12} using a greedy search method which is far more efficient than the semidefinite relaxation, and often yields more accurate solutions.

Despite the vast interest in phase retrieval, there has been little theoretical work on the fundamental limits of this problem.
One important question in this context is how many measurements are needed in order to ensure robust recovery of the input $x$, regardless of the specific recovery method used.
 Several recent works treat this problem.   Most of the papers discuss the case in which $x$ is a general input, namely, there is no sparsity (or other) constraint on $x$.
   The first result of this kind was obtained in \cite{BCE06}, where it is shown that with probability one $N=4n-2$ randomized equations are sufficient for recovery using a brute force (intractable) method, when there is no noise. However, it is not clear whether a stable recovery method exists with this number of measurements.
   In \cite{CL12,CSV12} the authors consider the case in which $a_i$ are real or complex vectors that are either uniform on the sphere of radius $\sqrt{n}$, or iid zero-mean Gaussian vectors with unit variance. Under these assumptions they show that on the order of $n$ measurements are needed in order to recover a generic $x$ using a semidefinite relaxation approach.
   In the presence of noise, it is shown in \cite{CL12} that one can find an estimate $\hat{x}$ satisfying
    \begin{equation}
\|\hat{x}-e^{i\phi}x\|_2 \leq C_0 \min \bl \|x\|_2,\frac{\|w\|_1}{N\|x\|_2}\br,
    \end{equation}
for some $\phi$, where $C_0$ is a constant and $w$ is the noise vector that is assumed to be bounded so that $\|w\|_1$ is finite.

The paper \cite{LV12} treats the case in which the input $x$ is $k$-sparse and $a_i$ are iid zero-mean normal vectors.
When there is no noise, they show that  in the real case $N \geq 4k-1$ measurements are needed for uniqueness and in the complex case, $N \geq 8k-2$ measurements are required.
They further prove that if $N$ is on the order of $k^2 \log n$ then the solution can be obtained using a sparse semidefinite relaxation approach as in \cite{SESE11,JOH12}.

Here we treat the real case and random measurements, using reasonable ensembles. Our methods may be extended to the complex case (again, using real, random measurements), but since the core of the problem is the real case, we will restrict our analysis to it.
For this setting, we develop conditions leading to stable uniqueness, namely, conditions that ensure a unique solution can be found in a stable way. We do not restrict ourselves to a certain class of inputs, but rather allow for general input sets $T$ which include, as special cases, $T=\RR^n$ and the class of sparse vectors. It turns out that to ensure stable uniqueness for a given set of signals, a very natural notion of complexity of the set determines the number of data points required. For example, we show that for $k$-sparse vectors,  $O(k\log(n/k))$ measurements are needed for stability. This result is better by a factor of $k$ than the estimate from \cite{LV12} that guarantees recovery (without noise) using a semidefinite approach.
When $x$ can be any vector in $\R^n$, we show that $O(n)$ measurements suffice, which is also the bound derived in \cite{CL12} for recovery using semidefinite relaxation.

It turns out that the natural complexity parameter for this problem is the same as the one used for analyzing stability in linear measurements, as we will discuss in Section~\ref{sec:linear}. Thus, in a rather general sense, the number of measurements required for stable recovery in the quadratic setting we treat here is of the same order of magnitude as the one needed to ensure stability under linear sampling. In that sense there is no substantial price to be paid for not knowing the phase of the measurements, and for very general choices of input sets $T$.

The second main result of this article deals with the noisy phase retrieval problem.
More specifically, we consider recovering an input $x$ in a set $T$ from noisy measurements of the form (\ref{eq:meas}).
A straightforward approach is to seek the value of $x$ that minimizes the empirical risk (or a least-squares approach). Since this leads to a nonconvex problem, finding its global solution is in general not possible. Nonetheless, we show that if one can find a value $\hx$ for which the empirical risk is bounded by a given, computable constant (which depends on the set $T$), then $\|\hx-x_0\|_2  \|\hx+x_0\|_2$ is bounded above by an expression that once again depends on the complexity parameter of the set, and which converges to $0$ faster than $N^{-1/2+\delta}$ for any $\delta>0$. Here $x_0$ is the true (unknown) input.
The complexity parameter that determines the rate in the noisy setting is essentially the same as in the stability analysis. Moreover, the resulting sample complexity is of the same order of magnitude as in the linear case in the examples of sets $T$ we consider.
An exact formulation of both main results is presented in the next section.

An important practical conclusion from this analysis is that although the squared-error in the case of nonlinear measurements as in (\ref{eq:meas}) cannot be minimized directly, it is sufficient to find a point for which the error is bounded by a known constant. Thus, one may use any desired recovery algorithm and check whether the solution $\hx$ satisfies the bound. For this purpose, methods such as those developed in \cite{SBE12} are advantageous since they allow for arbitrary initial points. As different initializations lead to different choices of $\hx$, the algorithm can be used several times until an appropriate value of $\hx$ is found. Our theoretical results ensure that such an $\hx$ is sufficiently close to $x_0$ or to $-x_0$ if enough measurements are used. In particular,
for $k$-sparse vectors one can guarantee stable recovery from $O(k\log(n/k)\log k)$ noisy measurements, and when $x$ can be any vector in $\R^n$, $O(n \log n)$ noisy measurements suffice.

The reminder of the article is organized as follows.
The problem and the main results are formulated in Section~\ref{sec:main}.
Stability results in the noise-free setting are developed in Section~\ref{sec:noisef}, while
the noisy setting is treated in Section~\ref{sec:noise}. In Section~\ref{sec:linear} the relation between the results in the quadratic case and those in the linear setting are discussed.

Throughout the article we use the following notation. The statistical expectation is denoted by $\E$, and if the probability space is a product space $(\Omega \times \Omega^\prime, \mu \otimes \mu^\prime)$, $\E_{\mu}$ and $\E_{\mu^\prime}$ are the conditional expectations. If $X$ is a random variable, then $\|X\|_{L_p} = (\E |X|^p)^{1/p}$.
The relation $a \sim b$ means that $a$ is equal to $b$ up to absolute multiplicative constants, i.e., that there are $c$ and $C$, independent of $a,b$ or any other parameters of the problem, for which $ca \leq b \leq Ca$. The inequality $a \lesssim b$ means that $a \leq C b$ for some constant $C$. We use $a \lesssim_{L,\gamma} b$ to denote the fact that the constant $C$ depends only on $L$ and $\gamma$.

\section{Problem Formulation and Main Results}
\label{sec:main}

Suppose one is given measurements $y_i$ as in (\ref{eq:meas}). Let $s$ be a vector in $\RR^N$, set $\phi(s)$ to be the length-$N$ vector with elements $|s_i|^2$ and put $Ax=(\inner{a_i}{x})_{i=1}^N$. With this notation, (\ref{eq:meas}) can be written as
\begin{equation}
\label{eq:model}
y=\phi(Ax)+w.
\end{equation}
Our goal is to study conditions under which stable recovery is possible irrespective of the specific recovery method used, and to develop guarantees that ensure that empirical minimization or approximate empirical minimization (namely, least-squares recovery) lead to an estimate $\hat{x}$ that is close to $x$ in a squared-error sense.

\subsection{Assumptions on $x$ and $a$}

We assume throughout that $x$ lies in a subset $T$ of $\RR^n$, which can be arbitrary. It is natural to expect that the number of measurements needed for stable recovery or for noisy recovery depend on the set $T$, though the way in which it depends on $T$ is not obvious. Here, we prove that this number is a function of a natural complexity parameter of $T$, in an estimate that is sharp in the stable recovery problem, and sharp up to logarithmic factors in the noisy recovery problem.

The assumption on the measurement vectors $a_i$ is that they are independent, and distributed according to a probability measure $\mu$ on $\R^n$ that is isotropic and $L$-subgaussian \cite{LAMA,VW,BK00}:
\begin{Definition} \label{def-iso-subgaussian}
Let $\mu$ be a probability measure on $\R^n$ and let $a$ be distributed according to $\mu$. The measure $\mu$ is isotropic if for every $t \in \R^n$, $\E|\inr{a,t}|^2 = \|t\|_{2}^2$. It is  $L$-subgaussian if for every $t \in \R^n$ and every $u \geq 1$, $Pr(|\inr{a,t}| \geq Lu\|\inr{t,a}\|_{2}) \leq 2\exp(-u^2/2)$.
\end{Definition}

Among the examples of isotropic, $L$-subgaussian measures on $\R^n$ for a constant $L$ that is independent of the dimension $n$ are the standard Gaussian measure, the uniform measure on $\{-1,1\}^n$ and the volume measure on the ``correct" multiple of the unit ball of $\ell_p^n$ for $2 \leq p \leq \infty$ (that is, the volume measure on  $c_nn^{1/p}B_p^n$, where $c_n \sim 1$). Also, if $X$ is a mean-zero, variance $1$ random variable that satisfies $Pr(|X| \geq Lu) \leq 2\exp(-u^2/2)$, then a vector of iid copies of $X$ is isotropic and $cL$ subgaussian, for a suitable absolute constant $c$.

More generally, we have the following:
\begin{Definition} \label{def:generel-subgaussian}
If $F$ is a class of functions on a probability space $(\Omega,\mu)$, then it is $L$-subgaussian if for every $f,h \in F \cup \{0\}$ and every $t \geq 1$
$$
Pr\left( |f-h|(X) \geq t L\|f-g\|_{2} \right) \leq 2\exp(-t^2/2),
$$
where $X$ is distributed according to $\mu$.
\end{Definition}

It is standard to verify if $\mu$ is an $L$-subgaussian measure on $\R^n$ then every class of linear functionals on $\R^n$ is $L$-subgaussian (see, e.g. \cite{LAMA}).

Our goal is to study when the mapping $\phi(Ax)$ is both invertible and stable first, when  $w=0$ (the noise-free case) and second, in the presence of noise.

\subsection{Stability Results}

We begin in Section~\ref{sec:noisef} by treating the noise-free setting. Since one is given only the absolute values of $Ax$, it is impossible to distinguish $x$ and $-x$. Therefore, uniqueness will always be up to the sign of $x$.  If $\phi(Ax)$ is an invertible stable mapping, it is natural to expect that for any $s$ and $t$ for which $s \neq t$ and $s \neq -t$, $\phi(As)$ is far enough from $\phi(At)$ in some sense; here we consider the $\ell_1$ sense.
\begin{Definition}
\label{def:stable}
The mapping $\phi(Ax)$ is stable with a constant $C$ in a set $T$ if for every $s,t \in T$,
\begin{equation}
\|\phi(A t) -\phi(A s)\|_{1} \geq C \|s-t\|_{2} \|s+t\|_{2}.
\end{equation}
\end{Definition}
Note that stability in a set is a much stronger property than invertibility. Indeed, for the latter it suffices that if $s \not = \pm t$ then $\|\phi(A t) -\phi(A s)\|_{1} >0$, but without any quantitative estimate on the difference.

The $\ell_1$ norm, used on the left-hand side, is the natural way of measuring distances for the quadratic function $\phi$, if one wishes to compare the results with the linear case, in which the $\ell_2$ distance is used (see Section \ref{sec:linear} for more information).
Using the $\ell_1$ distance also has a technical advantage, as it simplifies the analysis considerably. Measure distances based on other $\ell_p$ norms lead to processes that are much harder to control, since higher powers emphasize the ``unbounded" or ``peaky" parts of a random variable, and make concentration around the mean much harder.

To formulate our stability result, let us define the main complexity parameter required. For $T \subset \R^n$, define
\begin{eqnarray}
\label{eq:tpm}
T_-& = & \left\{\frac{t-s}{\|t-s\|_2} \ : \ t,s \in T, \ t \not = s \right\}, \nonumber \\ T_+ & = & \left\{\frac{t+s}{\|t+s\|_{2}} \ : \ t,s \in T, \ t \not = - s \right\}.
\end{eqnarray}
Let $(g_i)_{i=1}^n$ be independent Gaussian random variables, that have mean zero and variance $1$. Set
\begin{equation}
\label{eq:rhoc}
E=\max\left\{\E \sup_{v \in T_-} \sum_{i=1}^n g_i v_i, \ \E \sup_{w \in T_+} \sum_{i=1}^n g_i w_i\right\}
\end{equation}
and put
\begin{equation}
\label{eq:rho}
\rho_{T,N} = \frac{E}{\sqrt{N}} + \frac{E^2}{N}.
\end{equation}
Throughout this section, we will refer to $\rho_{T,N}$ as the complexity measure of $T$.

The main result in the noise free case is the following:
\begin{Theorem} \label{thm:main-random}
For every $L \geq 1$ there exist constants $c_1,c_2$ and $c_3$ that depend only on $L$ for which the following holds. Let $\mu$ be an isotropic, $L$-subgaussian measure.
Then, for $u \geq c_1$, with probability at least $1-2\exp(-c_2u^2\min\{N,E^2\})$, for every $s,t \in T$,
$$
\| \phi(As)-\phi(At) \|_{1} \geq \|s-t\|_{2} \|s+t\|_{2} \left( \kappa(s-t,s+t) - c_3u^3\rho_{T,N}\right),
$$
where for every $v,w \in \R^n$,
\begin{equation}
\label{eq:kappa}
\kappa(v,w)=\E|\inr{a,v/\|v\|_{2}} \inr{a,w/\|w\|_{2}}|.
\end{equation}
\end{Theorem}
To put Theorem \ref{thm:main-random} in the right perspective, one has to obtain lower bounds on $\kappa(s-t,s+t)$ and upper bounds on $\rho_{T,N}$. Since the latter depends on the number of measurements $N$, its behavior provides insight into the number of measurements that are needed for stability.

The value of $\kappa(v,w)$ may be bounded using several methods, as we will explain in Section \ref{sec:kappa}. One natural example in which $\inf_{v,w \in S^{n-1}} \kappa(v,w)$ is bounded from below (where $S^{n-1}$ is the unit Euclidean sphere in $\R^n$), is when $a$ satisfies a small-ball assumption, namely, that for every $t \in \R^n$ and every $\eps>0$,
\begin{equation} \label{eq:small-ball}
Pr(|\inr{a,t}| \leq \|t\|_2 \eps) \leq c\eps.
\end{equation}
It turns out that if \eqref{eq:small-ball} holds, then $\inf_{v,w \in S^{n-1}} \kappa(v,w) \geq c_1$, where $c_1$ only depends on the constant $c$ in \eqref{eq:small-ball}. This assumption is satisfied for a large family of measures, such as the Gaussian measure on $\R^n$ (see Section~\ref{sec:kappa} for more details).

As for $\rho_{T,N}$, we will show, for example, that if $T$ is the set of $k$-sparse vectors in $\R^n$, then $\rho_{T,N} \lesssim \sqrt{
k\log(en/k)/N}$. Hence, under \eqref{eq:small-ball}, since $\sloppy \inf_{v,w \in S^{n-1}} \kappa(v,w) \geq c_1$, it suffices to select $N$ large enough to ensure that $c_2 u^3 \rho_{T,N} \leq c_1/2$ to obtain a stability result. This leads to the following estimate:
\begin{Corollary} \label{cor-simple-intro}
For every $L,c>0$ there exist absolute constants $c_1$, $c_2,c_3$ and $c_4$ for which the following holds. Let $T$ be the set of $k$-sparse vectors in $\R^n$, set $\mu$ to be an isotropic, $L$-subgaussian measure, and assume that $a$ is distributed according to $\mu$. If $a$ satisfies \eqref{eq:small-ball} with constant $c$, then for $u>c_1$ and $N \geq c_2 u^3k\log(en/k)$, with probability at least $1-2\exp(-c_3u^2k\log(en/k))$, for every $s,t \in T$
$$
\|\phi(As)-\phi(At)\|_1 \geq c_4\|s-t\|_2 \|s+t\|_2.
$$
In particular, the result is true for a random Gaussian matrix $A$, where $c_1,c_2,c_3$ and $c_4$ are absolute constants.
\end{Corollary}

Interestingly, it can be shown that in the case of linear measurements, stable recovery is guaranteed as long as $N \sim k\log(en/k)$. Thus, the number of measurements needed for stable recovery in the nonlinear and linear settings is the same up to multiplicative constants -- at least for ensembles that have a well behaved $\inf_{v,w \in S^{n-1}} \kappa(v,w)$. As mentioned in the introduction, this observation is not a coincidence and will be explained in more detail in Section~\ref{sec:linear}.

In Section~\ref{sec:examples} we study other choices of $T$, and the number of measurements needed in order to guarantee stability.

\subsection{Noisy Recovery Results}

Section~\ref{sec:noise} is devoted to the case in which the measurements are contaminated with iid noise. The goal is to find a point $\hat{x}$ for which $\|\hat{x}-x_0\|_2 \|\hat{x}+x_0\|_2$ is small, using the data $(a_i,y_i)_{i=1}^N$ and the fact that $y$ is generated according to (\ref{eq:meas}) for some $x_0 \in T$.

A natural approach is to recover $x_0$ from $y$ by minimizing the empirical risk:
\begin{equation}
\label{eq:erm}
\min_x \ell_x=\min_x \frac{1}{N}\sum_{i=1}^N \left|y_i-|\inner{a_i}{x}|^2\right|^p,
\end{equation}
for some $p>1$.
The objective in (\ref{eq:erm}) is not convex, and therefore it is not clear how to find the value $x_0$ minimizing (\ref{eq:erm}).
Fortunately, we prove that in order to find an estimate $\hx$ close to $x_0$ one does not need to strictly minimize (\ref{eq:erm}). Instead, it is sufficient to find a point $\hx$ for which the empirical risk $\ell_\hx$ is small enough.

\begin{Definition}
Given a set $T \subset \R^n$, let
\begin{equation}
\label{eq:lt}
\ell(T)=\E \sup_{t \in T} |\sum_{i=1}^n g_i t_i|,
\end{equation}
be the Gaussian complexity of $T$, where $g_1,...,g_n$ are iid standard Gaussian variables, and put
\begin{equation}
\label{eq:dt}
d(T)=\sup_{t \in T} \|t\|_2.
\end{equation}
\end{Definition}

From the geometric point of view, $\ell(T)$ measures the best correlation (or width) of $T$ in a random direction generated by the random vector $G=(g_1,...,g_n)$. This parameter appears in many different areas of mathematics, and is also essential in the study of compressed sensing problems (see, for example \cite{LAMA}). We refer the reader to the books \cite{LT,Tal,Pis,MS} for more information on this parameter and for methods of computing it. For example, it is well known that $\ell(T)$ can be bounded from above and below (with a possible $\sqrt{\log n}$ gap between the upper and lower bounds) using of the $\ell_2$ covering numbers of $T$.

Suppose that
for a given $1<p \leq 2$, and $u \geq 1$ (which will later on govern our probability estimates), one produces $\hx$ satisfying
\begin{equation} \label{eq:hax-tt}
\frac{1}{N}\sum_{i=1}^N \left| |\inr{a_i,\hx}|^2 - y_i \right|^p \leq \E|w|^p + u\left(Q_{T,N,W}-\frac{\||w|^p\|_{\psi_1}}{\sqrt{N}}\right).
\end{equation}
Here $\||w|^p\|_{\psi_\alpha}$, $1 \leq \alpha \leq 2$ is a measure of the decay properties of the noise, and will be defined formally in (\ref{sec:noised}),
$Q_{T,N}$ is a complexity measure similar to $\rho_{T,N}$ (defined formally in (\ref{eq:qtn})), and $Q_{T,N,W}=Q_{T,N} +\||w|^p\|_{\psi_1}/\sqrt{N}$. Our main result shows that, with high probability, such a point $\hx$ is close to either $x_0$ or to $-x_0$.
To find an appropriate $\hx$, it is possible, for example, to use the greedy method of \cite{SBE12} with different starting points and stop once a solution that satisfies the bound is found.
\begin{Theorem} \label{thm:noisy-maint}
For every $\kappa>0$ and every $L \geq 1$ there exists constants $c_1,c_2,c_3$ and $c_4$ that depend only on $L$ and $\kappa$, for which the following holds. Let $a$ be distributed according to an isotropic, $L$-subgaussian measure, and assume that $\kappa_T \geq \kappa$ where $\kappa_T=\inf_{s,t\in T} \kappa(s,t)$. Assume further that $\|w\|_{\psi_2} < \infty$.
For every integer $N$ set
$$
\beta_N = \max\{c_1\left((\|w\|_{\psi_1} + d^2(T)) \log N + \ell^2(T)\right),e\}
$$
and
$$p=1+1/\log \beta_N.$$ Let $\hx$ be chosen to satisfy \eqref{eq:hax-tt}.
Then, for $u \geq c_2$, with probability at least $1-2\exp(-c_3u^{1/3})$,
$$
\|x_0-\hx\|_2  \|x_0+\hx\|_2 \leq c_4 (uQ_{T,N,W})^{1/p} \sqrt{\log \beta_N}.
$$
\end{Theorem}
Theorem~\ref{thm:noisy-maint} shows that stable recovery is possible if $N \gtrsim Q_{T,N,W}^{1/p} \sqrt{\log \beta_N}$.
 In particular, for $k$-sparse vectors, stable recovery is possible from $O(k\log(n/k)\log k)$ noisy measurements (this estimate is off only by a $\log k$
  factor from the
optimal estimate in the linear case), and
when $x_0$ can be any vector in $\R^n$, $O(n \log n)$ noisy measurements suffice.

\subsection{Technical Tool}
The main technical tool needed in the proof of both main results is a general estimate on properties of empirical processes indexed by $\{f h : f \in F, \ h \in H\}$. Although the result is true in a far more general situation than needed here, for the sake of simplicity we will present it only in the cases required. We refer the reader to \cite{Men-orc} for the more general statement and precise results.

In the cases considered here, $F$ and $H$ are classes of linear functionals or of absolute values of linear functionals on $\R^n$, which is endowed with an isotropic, $L$-subgaussian probability measure $\mu$. For the stability result $F=\{ |\inr{t, \cdot}| : t \in T_+\}$ and $H=\{ |\inr{t,\cdot}| : t \in T_-\}$, while in the noisy case, $F=\{\inr{t-t_0,\cdot} : t \in T\}$ and $H=\{\inr{t+t_0,\cdot} : t \in T\}$.
In both scenarios, the two indexing sets are denoted by $T_1,T_2 \subset \R^n$.

\begin{Theorem} \label{thm:main-emp-est} \cite{Men-orc}
For every $L \geq 1$ there are constants $c_1,c_2,c_3$ and $c_4$ that depend only on $L$ and for which the following hold. Let $T_1,T_2 \subset \R^n$ of cardinality at least $2$ and set $F$ and $H$ to be the corresponding classes as above, respectively. Assume without loss of generality that $\ell(T_1)/d(T_1) \geq \ell(T_2)/d(T_2)$. Then, for every $u \geq c_1$, with probability at least
$$
1-2\exp\left(-c_2u^2 \min\{N, (\ell(T_1)/d(T_1))^2\}\right),
$$
$$
\sup_{f \in F, h \in H} \left|\frac{1}{N} \sum_{i=1}^N f(a_i)h(a_i) - \E f h \right| \leq c_3 u^3 \left(d(T_2) \frac{\ell(T_1)}{\sqrt{N}} + \frac{\ell^2(T_1)}{N}\right).
$$
In particular, for every $q \geq 2$,
$$
\left\|\sup_{f \in F, h \in H} \left|\frac{1}{N} \sum_{i=1}^N f(a_i)h(a_i) - \E f h \right| \right\|_{L_q}
\leq  c_4 q^{3/2}\left(d(T_2) \frac{\ell(T_1)}{\sqrt{N}} + \frac{\ell^2(T_1)}{N}\right).
$$
\end{Theorem}

\section{Stability Results}
\label{sec:noisef}
In this section we present the proof of Theorem~\ref{thm:main-random}, followed by estimates on the values of $\kappa(v,w)$ and $\rho_{T,N}$ appearing in the theorem.

\subsection{Proof of Theorem~\ref{thm:main-random}}

Observe that
\begin{equation}
\label{eq:equiv1}
\|\phi(A t) -\phi(A s)\|_{1} = \sum_{i=1}^N \left| |\inr{a_i,t}|^2 - |\inr{a_i,s}|^2 \right| = \sum_{i=1}^N |\inr{a_i, s-t} \inr{a_i,s+t}|.
\end{equation}
Therefore, to establish the desired stability result, it suffices to show that
$$
\inf_{\{s,t \in T, \ s \not = \pm t\}} \frac{\|\phi(A t) -\phi(A s)\|_{1}}{\|s-t\|_{2} \|s+t\|_{2}} \geq C,
$$
i.e., that
\begin{equation}
\label{eq:stable-alt}
\inf_{\{s,t \in T, \ s \not = \pm t\}} z_{t,s}\geq \frac{C}{N},
\end{equation}
where
\begin{equation}
\label{eq:defz}
z_{t,s}=\frac{1}{N} \sum_{i=1}^N \left|\inr{a_i, \frac{s-t}{\|s-t\|_2}}  \inr{a_i,\frac{s+t}{\|s+t\|_2}}\right|.
\end{equation}

Since $\kappa(s-t,s+t)=\E z_{t,s}$, if $\kappa(s-t,s+t)$ is very small, then a random selection of $a_i$ is unlikely to lead to \eqref{eq:stable-alt}. Therefore, a reasonable pre-requisite for a stability result is that $\inf_{ s \not \pm t, \ s,t \in T} \kappa(s-t,s+t)$ is bounded away from zero. Indeed, with this assumption, one may obtain a stability result.

\begin{Proposition}
 \label{prop:main-random}
For every $L \geq 1$ there exist constants $c_1,c_2$ and $c_3$ that depend only on $L$ for which the following holds. Let $\mu$ be an isotropic, $L$-subgaussian measure on $\R^n$ and set $a$ to be a random vector distributed according to $\mu$.
Then, for $u \geq c_1$, with probability at least $1-2\exp(-c_2u^2\min\{N,E^2\})$, for every $s,t \in T$,
$$
z_{s,t} \geq \left(\kappa(s-t,s+t) -c_3u^3 \rho_{T,N}\right)\|s-t\|_2 \|s+t\|_2.
$$
\end{Proposition}
\begin{proof}
Observe that
$$
\sup |z_{t,s} - \kappa(s-t,s+t)| = \sup_{v \in T_+, \ w \in T_ -} \left| \frac{1}{N} \sum_{i=1}^N \left|\inr{a_i,v}  \inr{a_i,w}| - \E |\inr{a_i,v} \inr{a_i,w}\right| \right|.
$$
By Theorem~\ref{thm:main-emp-est} for $F=\{|\inr{v, \cdot}| : v \in T_+\}$ and $H=\{|\inr{w,\cdot}| : w \in T_-\}$, it follows that  if $N \geq c_1 E^2$ and $u \geq c_2$ then with probability at least $1-2\exp(-c_3u^2 E^2)$,
\begin{equation} \label{eq:cov-est}
\sup_{v \in T_-, \ w \in T_+} \left| \frac{1}{N} \sum_{i=1}^N \left|\inr{a_i,v} \inr{a_i,w}| - \E |\inr{a_i,v} \inr{a_i,w}\right| \right| \leq c_4u^3 \rho_{T,N}.
\end{equation}
The claim now follows immediately from the definition of $z_{s,t}$ and of $\kappa(s-t,s+t)$.
\end{proof}

\subsection{Computing $\kappa$ and $\rho_{T,N}$}
\label{sec:examples}

\subsubsection{Bounding $\rho_{T,N}$}
It is well known that if $T \subset \R^n$ then $\ell(T)$ (and therefore, $\rho_{T,N}$ as well) is determined by the Euclidean metric structure of $T$. This is the outcome of the celebrated majorizing measures/generic chaining theory (see the books \cite{LT,Dud-book,Tal} for a detailed exposition on this topic). In the examples we present here, the following estimate, which is, in general, suboptimal, suffices.
\begin{Definition} \label{def:cover}
Let $(T,d)$ be a compact metric space. For every $\eps>0$, let $N(T,d,\eps)$ be the smallest number of open balls of radius $\eps$ needed to cover $T$. The numbers $N(T,d,\eps)$ are called the $\eps$-covering numbers of $T$ relative to the metric $d$.
\end{Definition}
Given $T \subset \R^n$, set $N(T,\eps)=N(T,\| \ \|_2,\eps)$, i.e., the covering numbers relative to the Euclidean metric.
\begin{Proposition} \label{prop:rho-cover}
There exist absolute constants $c$ and $C$ for which the following holds.
If $T \subset \R^n$ then
$$
c \sup_{\eps>0 }\eps \sqrt{\log N (T,\eps)}\leq  \ell(T)  \leq C\int_0^{d(T)} \sqrt{\log N(T,\eps)} d\eps.
$$
\end{Proposition}
The upper bound is due to Dudley \cite{Dud67} and the lower to Sudakov \cite{Sud}. The proof of both bounds may be found, for example, in \cite{LT,Pis,Dud-book}.

It is straightforward to verify that the gap between the upper and lower bounds in Proposition \ref{prop:rho-cover} is at most $\sim \sqrt{\log n}$, and in all the examples we study below, the resulting estimate is sharp.

\subsubsection{Bounding $\kappa$} \label{sec:kappa}
Here, we present two simple methods for bounding $\inf_{v,w \in S^{n-1}} \kappa(v,w)$ from below. These methods are not the only possibilities by which one may obtain such a bound; rather, they serve as an indication that the assumption on $\kappa$ is less restrictive than may appear at first glance.

Recall the small ball assumption: for every $t \in \R^n$ and every $\eps>0$, $Pr( |\inr{a,t}| \leq \eps \|t\|_2 ) \leq c\eps$.
\begin{Lemma} \label{lemma:small-ball}
If $a$ satisfies the small ball assumption with constant $c$ then
$$
\inf_{v,w \in S^{n-1}} \kappa(v,w) \geq \kappa,
$$
where $\kappa$ depends only on $c$.
\end{Lemma}

\proof
Consider $\eps$ for which $c\eps \leq 1/4$. Then for every $v \in S^{n-1}$, there is an event of measure at least $3/4$ on which $|\inr{v,a}| \geq \eps$. Hence, for two fixed vectors $v,w \in S^{n-1}$,
$$
Pr( \{|\inr{v,a}| \geq \eps\} \cap \{  |\inr{w,a}| \geq \eps \} ) \geq 1/2,
$$
and thus $\E |\inr{v,a} \inr{w,a}| \geq \eps^2/2$.
\endproof

This type of small ball property is true in many case. The simplest example is the standard gaussian measure on $\R^n$. Indeed, if $a=(g_1,...,g_n)$ then $|\inr{a,w}|$ is distributed as $\|w\|_2 |g|$ and the small ball property follows immediately by applying the $L_\infty$ estimate on the density of $g$.

A more general example is based on the notion of log-concavity. A measure $\mu$ on $\R^n$ is called log-concave if for every nonempty, Borel measurable sets $A,B \subset \R^n$, and any $0 \leq \lambda \leq 1$, $\mu(\lambda A + (1-\lambda)B) \geq \mu^\lambda(A) \mu^{1-\lambda}(B)$, where $\lambda A+ (1-\lambda)B = \{\lambda a + (1-\lambda)b : a \in A, b \in B\}$. It is well known that $\mu$ is a log-concave measure if and only if it has a density of the form $\exp(\phi)$ for a concave function $\phi:\R^n \to \R$.

The following lemma is standard (see e.g. \cite{Gia,Bob,MP}).

\begin{Lemma} \label{lemma:log-concave}
There exists an absolute constant $c$ for which the following holds.
Let $a$ be distributed according to an isotropic, symmetric, log-concave measure. Then for every $\theta \in S^{n-1}$, $\inr{a,\theta}$ is distributed according to an isotropic, symmetric, log-concave measure on $\R$. Also, if $f_\theta$ is the density of $\inr{a,\theta}$ then $\|f_\theta\|_{\infty} \leq c$.
\end{Lemma}
The desired small-ball estimate clearly follows from the lemma, since
$$
Pr(|\inr{a,\theta}| \leq \eps) = \int_{-\eps}^\eps f_\theta(t)dt \leq 2c\eps.
$$
Among the family of log-concave measures are volume measures on convex symmetric bodies (i.e. measures that have a constant density on the body and zero outside the body). Moreover, it can be shown (see, e.g. \cite{Gia}), that for every convex body $K \subset \R^n$ there is an invertible linear operator $G$ for which the volume measure on $GK$ is also isotropic.

Another example of log-concave measures on $\R^n$ are product measures of log-concave measures on $\R$. If $X$ is a real valued, symmetric,  log-concave random variable (i.e. with a log-concave density) with variance one, and $X_1,...,X_n$ are iid copies of $X$, then $a=(X_1,...,X_n)$ is an isotropic log-concave measure on $\R^n$.
Standard examples for log-concave measures on $\R$ are those with a density proportional to $\exp(-c_p|t|^p)$ for $p \geq 1$. And, of course, any measure with density proportional to $\exp(\phi)$ for a convex function $\phi$ is log-concave.

The second method, which we only outline, is based on the Paley-Zygmund argument.
\begin{Lemma} \label{Lemma:Paley-Zygmund} \cite{GdlP}
Let $Z$ be a random variable, set $0<p<q$ and put $c_{p,q}=\|Z\|_{L_p}/\|Z\|_{L_q}$. Then, for every $0 \leq \lambda \leq 1$,
$$
Pr(|Z| > \lambda \|Z\|_{L_p}) \geq [(1-\lambda^p)c_{p,q}^p]^{q/q-p},
$$
and in particular,
$\E|Z| \geq c_1 \|Z\|_{L_p}$, where $c_1$ depends only on $p,q$ and $c_{p,q}$.
\end{Lemma}
We will use the lemma for $p=2$ and $q>2$. Assume that $(X_i)_{i=1}^n$ are iid copies of a symmetric, variance $1$ random variable and set $a=(X_1,...,X_n)$. If $v,w \in S^{n-1}$, then a straightforward computation shows that
\begin{align} \label{eq:moment-computation-1}
\E |\inr{a,v} \inr{a,w}|^2 = & \sum_{i \not = j} v_i^2 w_j^2 + 2 \sum_{i \not = j} (v_i w_i) (v_j w_j) + \E X^4 \sum_{i=1}^N v_i^2 w_i^2 \nonumber
\\
= & \sum_{i \not = j} \left(v_i^2 w_j^2+ (v_i w_i)  (v_j w_j)\right) + \inr{v,w}^2 + (\E X^4 -1)\sum_{i=1}^n v_i^2 w_i^2.
\end{align}
Using the fact that $\|v\|_2=\|w\|_2=1$, \eqref{eq:moment-computation-1} reduces to
\begin{equation} \label{eq:moment-computation-2}
\E |\inr{a,v} \inr{a,w}|^2 =1+2\inr{v,w}^2-2\sum_{i=1}^n v_j^2w_j^2 + (\E X^4-1)\sum_{i=1}^n v_i^2 w_i^2.
\end{equation}

Consider two cases. First, if $\sum_{i=1}^n v_j^2w_j^2 \leq 1/10$, and since $\E X^4 \geq (\E X^2)^2=1$,
then by \eqref{eq:moment-computation-2},
$$
\E |\inr{a,v} \inr{a,w}|^2 \geq 1/2.
$$
On the other hand, if the reverse inequality holds, then using
$$
\sum_{i \not = j} \left(v_i^2 w_j^2+ (v_i w_i)  (v_j w_j)\right) + \inr{v,w}^2 = \sum_{i > j} (v_iw_j+v_jw_i)^2 + \inr{v,w}^2 \geq 0,
$$
and applying \eqref{eq:moment-computation-1},
$$
\E |\inr{a,v} \inr{a,w}|^2 \geq (\E X^4 -1)\sum_{i=1}^n v_i^2 w_i^2 \geq (\E X^4 -1)/10.
$$
\begin{Corollary} \label{cor-small-ball-moments}
Let $X$ be a symmetric, variance $1$ random variable, with a finite $L_{2q}$ moment for some $q>2$. If $a=(X_1,...,X_n)$ then
$$
\inf_{v,w \in S^{n-1}} \kappa(u,v) \geq c(\E X^4 -1)^{1/2},
$$
where $c$ depends on $q$ and on $\|X\|_{L_{2q}}$.
\end{Corollary}
Observe that the two assumptions we make are not very restrictive, since we assume throughout that $a$ is isotropic and $L$-subgaussian. Hence, if $a=(X_1,...,X_n)$, then  $\|X\|_{L_q} \leq Lq$ for every $q \geq 2$ (see Section \ref{sec:prelim}). Also note that for any random variable $X$, $\E X^4 \geq (\E X^2)^2=1$, so that the square-root is well defined.

\proof
Assume that $X \in L_{2q}$ for some $q>2$. Observe that if $v \in S^{n-1}$ then for every $2 \leq r \leq 2q$, $\|\inr{a,v}\|_{L_r} \leq c_r \|X\|_{L_r}$. Indeed, by a Rosenthal type inequality (see, e.g. \cite{GdlP}, Section 1.5),
$$
\|\sum_{i=1}^n v_i X_i \|_{L_r} \leq c_r \max \left\{ \left(\sum_{i=1}^n v_i^2 \E X_i^2 \right)^{1/2}, \left(\sum_{i=1}^n v_i^r \E|X_i|^r \right)^{1/r} \right\}.
$$
Since $\|X\|_{L_r} \geq \|X\|_{L_2}$ and $\|v\|_r \leq \|v\|_2=1$, the claim follows.

Therefore, $\sup_{v \in S^{n-1}} \|\inr{a,v}\|_{L_{2q}} \leq c_q\|X\|_{L_{2q}}$, and thus,
$$
\| \inr{a,v} \inr{a,w} \|_{L_q} \leq \|\inr{a,v}\|_{L_{2q}} \|\inr{a,w}\|_{L_{2q}} \leq_q \|X\|_{L_{2q}}^2.
$$
Let $Z=\inr{a,v} \inr{a,w}$. Using the notation of Lemma \ref{Lemma:Paley-Zygmund}, $c_{2,q} \gtrsim (\E |X|^4 -1 )^{1/2}/\|X\|_{L_{2q}}^2$. Hence, for every $v,w \in S^{n-1}$, $\E |\inr{a,v} \inr{a,w}| \geq c(\E |X|^4 -1 )^{1/2}$, where $c$ depends only on $q$ and on $\|X\|_{L_{2q}}$, as claimed.
\endproof

Observe that if $\E X^4=1$, it is possible that $\E |\inr{a,v} \inr{a,w}|=0$, even for $v,w$ of the specific form one would like to control -- namely, $v=(s+t)/\|s+t\|_2)$ and $w=(s-t)/\|s-t\|_2$ for $s \not = \pm t$. Indeed, let $X$ be a symmetric, $\{-1,1\}$-valued random variable (and in particular, it is $L$-subgaussian as well). Let $(e_i)_{i=1}^n$ be the standard basis in $\R^n$ and set $s=e_1$, $t=e_2$. It is straightforward to verify that in this case, $\E |\inr{a,v} \inr{a,w}|=0$ with probability $1$, and therefore, the assumption on $\E X^4$ can not be relaxed if one is interested in a uniform bound.

\subsection{Examples}
Let us turn to a few special cases of Theorem~\ref{thm:main-random}. To that end, explicit expressions for $\rho_{T,N}$ are required for the sets of interest.

\subsubsection{Entire Space $T=\R^n$}
If $T=\R^n$ then $T_+=T_-=S^{n-1}$, where $S^{n-1}$ is the Euclidean unit sphere in $\R^n$. Therefore,
$$
E=\E \sup_{x \in S^{n-1}} \sum_{i=1}^n g_i x_i = \E(\sum_{i=1}^n g_i^2)^{1/2} \sim \sqrt{n},
$$
implying that
$$
\rho_{\R^{n},N} \lesssim \left(\sqrt{\frac{n}{N}}+\frac{n}{N}\right).
$$
\begin{Corollary} \label{cor:examples-entire-space}
For every $L \geq 1$ there are constants $c_1$, $c_2$ and $c_3$ that depend only on $L$ and for which the following holds.
If $\inf_{v,w \in S^{n-1}} \kappa(v,w) \geq \kappa$, $u \geq c_1$ and $N \geq c_2 u^3 n/\kappa^2$, then with probability at least $1-2\exp(-c_3u n)$, for every $s,t \in \R^n$,
$$
\|\phi(As)-\phi(At)\|_1 \geq \frac{\kappa}{2} \|s-t\|_2 \|s+t\|_2.
$$
\end{Corollary}
The corollary follows from the fact that with this choice of $N$, $cu^3\rho_{T,N}$ is proportional to $\kappa/2$.

When $\kappa$ is given by a constant, independent of the dimension $n$, Corollary~\ref{cor:examples-entire-space} implies that it is sufficient to choose $N \sim n$ to ensure stable recovery with high probability.

\subsubsection{Sparse Vectors}
Let $T=S_k$, the set of $k$-sparse vectors in $\R^n$, put
$U_k=\{x \in S^{n-1} : \|x\|_0 \leq k \}$ and observe that $T_+,T_- \subset U_{2k}$. Therefore,
$$
E=\E \sup_{x \in U_{2k}} \sum_{i=1}^n g_i x_i = \E\left(\sum_{i=1}^{2k} (g_i^*)^2\right)^{1/2},
$$
where $(v_i^*)_{i=1}^n$ is a monotone rearrangement of $(|v_i|)_{i=1}^n$. It is standard to check (see, e.g., \cite{GLMP}) that there is an absolute constant $c$ such that for every $1 \leq k \leq n/4$,
$$
\E \left(\sum_{i=1}^{2k} (g_i^*)^2 \right)^{1/2} \leq c \sqrt{k\log(en/k)}.
$$
Therefore,
$$
\rho_{S_k,N} \lesssim \left(\sqrt{\frac{k\log(en/k)}{N}}+ \frac{k\log(en/k)}{N}\right).
$$
\begin{Corollary} \label{cor:examples-sparse}
For every $L \geq 1$ there are constants $c_1$, $c_2$ and $c_3$ that depend only on $L$ and for which the following holds.
If $\inf_{v,w \in U_k} \kappa(v,w) \geq \kappa$, $u \geq c_1$ and $N \geq c_2 u^3 k\log(en/k)/\kappa^2$, then with probability at least $1-2\exp(-c_3u k\log(en/k))$, for every $s,t \in S_k$,
$$
\|\phi(As)-\phi(At)\|_1 \geq \frac{\kappa}{2} \|s-t\|_2 \|s+t\|_2.
$$
\end{Corollary}
When $\kappa$ is an absolute constant, Corollary~\ref{cor:examples-entire-space} implies that it is sufficient to choose $N \sim k\log(en/k)$ to ensure stable recovery with high probability.

\subsubsection{Finite Set}
Assume that $T$ is a finite set. Then $T_+,T_- \subset S^{n-1}$ are of cardinality at most $|T|^2$. A straightforward application of the union bound to each random variable $\sum_{i=1}^n v_i g_i$ shows that if $V \subset \R^n$ is a finite set, then
$$
\E \sup_{v \in V} \sum_{i=1}^n g_i v_i \lesssim \sqrt{\log |V|} d(V).
$$
Therefore, $E \lesssim \sqrt{\log |T|^2} \sim \sqrt{\log{|T|}}$, implying that
$$
\rho_{T,N} \lesssim \left(\sqrt{\frac{\log |T|}{N}} +  \frac{\log |T|}{N}\right).
$$
\begin{Corollary} \label{cor:examples-finite}
For every $L \geq 1$ there are constants $c_1$, $c_2$ and $c_3$ that depend only on $L$ and for which the following holds.
If $\inf_{v,w \in T_+} \kappa(v,w) \geq \kappa$, $u \geq c_1$ and $N \geq c_2 u^3 \log|T|/\kappa^2$, then with probability at least $1-2\exp(-c_3u \log|T|)$, for every $s,t \in T$,
$$
\|\phi(As)-\phi(At)\|_1 \geq \frac{\kappa}{2} \|s-t\|_2 \|s+t\|_2.
$$
\end{Corollary}
In this case, with constant $\kappa$, $N \sim \log |T|$ measurements ensure stable recovery with high probability.

\subsubsection{Block Sparse Vectors}
We next treat the case in which $T=S_k^d$ consists of block sparse vectors of size $d$ \cite{EKB10,EM09a}.
Let $(I_\ell)$, $\ell=1,...,n/d$ be a decomposition of $\{1,...,n\}$ to disjoint blocks of cardinality $d$. Set $W_k$ to be the set of vectors in the unit sphere, supported on at most $k$ blocks. Then $T_+,T_- \subset W_{2k}$, and it remains to estimate
$$
E=\E \sup_{v \in W_{2k}} \sum_{i=1}^n g_i v_i.
$$
\begin{Lemma} \label{lemma:ent-estimate-block}
There exist absolute constants $c_1$ and $c_2$ for which the following holds. For every $0<\eps<1/2$,
$$
\log N(W_k,\eps) \leq c_1 \left(k\log(en/(dk))+dk \log(5/\eps)\right).
$$
Therefore,
$$
\E \sup_{v \in V} \sum_{i=1}^n g_i v_i  \leq c_2 \sqrt{k} \left(\sqrt{\log(en/(dk))}+\sqrt{d}\right).
$$
\end{Lemma}
\proof
Let $I_J=\{i \in I_j, \ j \in J\}$ and observe that
$$
W_k = \bigcup_{\{J \subset \{1,...,n\}:|J|=k\}} S^{I_J},
$$
where for every $I \subset \{1,...,n\}$, $S^{I}$ is the Euclidean sphere on the coordinates $I$.
Clearly, there are at most $\binom{n/d}{k}$ such subsets $J$. Using a standard volumetric estimate (see, e.g., \cite{Pis,LAMA}), for every fixed set $J$ and every $\eps<1/2$, one needs at most $(5/\eps)^{d|J|}=(5/\eps)^{dk}$ Euclidean balls of radius $\eps$ to cover $S^{I_J}$. Therefore, for every $0<\eps<1/2$,
$$
\log N(W_k,\eps B_2^n) \lesssim  k\log(en/(dk)) + dk \log(5/\eps),
$$
as claimed.

The second part of the claim is an immediate consequence of Proposition~\ref{prop:rho-cover} and the fact that $N(T,\eps)$ is a decreasing function of $\eps$.
\endproof

\begin{Corollary} \label{cor:examples-block-sparse}
For every $L \geq 1$ there are constants $c_1$, $c_2$ and $c_3$ that depend only on $L$ and for which the following holds.
If $\inf_{v,w \in W_k} \kappa(v,w) \geq \kappa$, $u \geq c_1$ and $N \geq c_2 u^3 (k\log(en/(dk)) + dk)/\kappa^2$, then with probability at least $1-2\exp(-c_3u (k\log(en/(dk)) + dk))$, for every $s,t \in S_k^d$,
$$
\|\phi(As)-\phi(At)\|_1 \geq \frac{\kappa}{2} \|s-t\|_2 \|s+t\|_2
$$
\end{Corollary}
When $\kappa$ is constant we conclude that $N \sim k(\log(en/(kd)) + d)$ measurements are needed for stability. This result is consistent with that of \cite{EM09a} which shows that the same value $N$ ensures that a random Gaussian matrix satisfies the block restricted isometry constant.

\section{Noisy Measurements}
\label{sec:noise}

Next, consider the phase retrieval problem in the presence of noise. The goal is to find an estimate $\hx$ of the true signal $x_0$ that is close to $x_0$ (or $-x_0$) in a squared error sense.

Suppose that
\begin{equation}
\label{eq:measn}
y_i=|\inner{a_i}{x_0}|^2+w_i,\quad i=1,\ldots,N
\end{equation}
for some $x_0 \in T$. Let $a$ be an isotropic, $L$-subgaussian random vector and assume that the noise $w$ is independent of $a$, symmetric, and of reasonable decay properties, which will be specified in Assumption~\ref{asp:noise} below.
\begin{Question} \label{qu:noise}
Given $(a_i,y_i)_{i=1}^N$, combined with the information that the noisy data $y_i$ is generated by a point $x_0 \in T$ via (\ref{eq:measn}), is it possible to produce an estimate $\hx \in T$ for which $\|\hx-x_0\|_2 \|\hx+x_0\|_2$ is small?
\end{Question}
Note that the error is measured by the product $\|\hx-x_0\|_2 \|\hx+x_0\|_2$, since it is impossible to distinguish between $x_0$ and $-x_0$.

The answer to this question is affirmative, as shown in Theorem~\ref{thm:noisy-main}.

\subsection{Preliminaries: $\psi_\alpha$ Random Variables}
\label{sec:noised}

Throughout our analysis we assume that the noise $w$ decays properly.
In order to quantify this decay we rely on the notion of $\psi_\alpha$ random variables, which are defined below (see \cite{LT,VW,LAMA} as general references for properties of $\psi_\alpha$ random variables).
\label{sec:prelim}
\begin{Definition} \label{def-psi-alpha}
Let $X$ be a random variable. For $1 \leq \alpha \leq 2$ let
$$
\|X\|_{\psi_\alpha} = \inf \left\{ C>0 : \E \exp(|X/C|^\alpha) \leq 2\right\},
$$
and denote by $L_{\psi_\alpha}$ the set of random variables for which $\|X\|_{\psi_\alpha} < \infty$.
\end{Definition}
The $\psi_\alpha$ norm can be characterized using information on the tail of $X$. Indeed, there exists an absolute constant $c$, for which, if $t \geq 1$, then $Pr( |X| \geq t) \leq 2 \exp(-ct^\alpha/\|X\|_{\psi_\alpha}^\alpha)$. The reverse direction is also true, that is, if $Pr( |X| \geq t) \leq 2 \exp(-t^\alpha/A^\alpha)$, then $\|X\|_{\psi_\alpha} \leq c_1A$ for an absolute constant $c_1$.

It is well known that $\| \ \|_{\psi_\alpha}$ is a norm on $L_{\psi_\alpha}$, and that
$$
\|X\|_{\psi_\alpha} \sim \sup_{p \geq 1} \frac{\|X\|_{L_p}}{p^{1/\alpha}}.
$$
In other words,
\begin{equation} \label{eq:orlicz-moments}
\|X\|_{L_p} \lesssim \|X\|_{\psi_\alpha} p^{1/\alpha}, \quad \forall p>1.
\end{equation}

In the language of the previous section, $X$ is $L$-subgaussian if and only if $\|X\|_{\psi_2} \leq cL\|X\|_{L_2}$. Since the $\psi_\alpha$ norms have a natural hierarchy, it follows that if $X$ is $L$-subgaussian then
$$
\|X\|_{L_2} \lesssim \|X\|_{\psi_1} \leq \|X\|_{\psi_2} \leq cL\|X\|_{L_2}.
$$
Therefore, if $X$ is $L$-subgaussian and mean-zero then $\|X\|_{\psi_2} \sim_L \sigma_X$, where $\sigma_X$ is the standard deviation of $X$.

A straightforward application of the tail behavior of a $\psi_\alpha$ random variable implies that if $X_1,...,X_N$ are independent copies of $X$ and $t \geq 1$, then
$$
Pr\left( \max_{i \leq N} |X_i| \geq t \log^{1/\alpha} N \|X\|_{\psi_\alpha}\right) \leq 2\exp(-c_2t^{1/\alpha});
$$
hence,
\begin{equation} \label{eq:max-orlicz}
\|\max_{1 \leq i \leq N} X_i \|_{\psi_\alpha} \leq c_3 \|X\|_{\psi_\alpha} \log^{1/\alpha} N.
\end{equation}

From the definition of the $\psi_\alpha$ norm it is evident that if $\alpha=\beta/q$ then
\begin{equation}
\label{eq:psi-beta}
\| |X|^q \|_{\psi_\alpha} = \|X\|_{\psi_\beta}^q,
\end{equation}
and in particular, $X \in L_{\psi_\beta}$ for $\beta >1$ if and only if $|X|^\beta \in L_{\psi_1}$.

Although there are versions of the following theorem (and of Definition \ref{def-psi-alpha}) for any $0 < \alpha $, for the sake of simplicity, we shall restrict ourselves to the case $\alpha=1$, which is the setting needed in the proofs below.
\begin{Theorem} \label{thm:concentration-psi-1}
There exists an absolute constant $c_1$ for which the following holds. If $X \in L_{\psi_1}$ and $X_1,...,X_N$ are independent copies of $X$, then for every $t>0$,
$$
Pr \left( \left|\frac{1}{N}\sum_{i=1}^N X_i - \E X \right| > t\|X\|_{\psi_1} \right) \leq 2\exp(-cN \min\{t^2,t\}).
$$
\end{Theorem}
Combining Theorem~\ref{thm:concentration-psi-1} and \eqref{eq:psi-beta} leads to the following corollary:
\begin{Corollary}
\label{cor:ewp}
Let $p>1$ and assume that $w$ is a random variable for which $|w|^p \in L_{\psi_1}$ (or $w \in L_{\psi_p}$).
Then, with probability at least $1-2\exp(-ct)$,
$$
\left| \frac{1}{N} \sum_{i=1}^N |w_i|^p - \E |w|^p \right| \leq \|w\|_{\psi_p}^p \sqrt{\frac{t}{N}}.
$$
\end{Corollary}
The corollary follows immediately from Theorem \ref{thm:concentration-psi-1} by taking $t^\prime = \sqrt{t/N}$ for $0<t<N$, and since $\||w|^p\|_{\psi_1}=\|w\|_{\psi_p}^p$.

We will also be interested in decay properties of the random variable $\sup_{t \in T} |\inr{X,t}|$ for a set $T \subset \R^n$.
If $\mu$ is an isotropic, $L$-subgaussian measure on $\R^n$, one has the following (see, e.g. \cite{Men-GAFA}).
\begin{Theorem} \label{thm:monotone-psi-1}
For every $L>1$ there exist constants $c_1$, $c_2, c_3$ and $c_4$ that depend only on $L$ and for which the following holds. If $u \geq c_1$, then with probability at least $1-2\exp(-c_2 u \log N)$,
$$
\max_{1 \leq i \leq N} \sup_{t \in T} |\inr{a_i,t}|^2 \leq c_3u\left(\ell^2(T)+d^2(T) \log N\right),
$$
where $\ell(T)$ and $d(T)$ are defined by \eqref{eq:lt} and \eqref{eq:dt}.
In particular,
$$
\| \max_{1 \leq i \leq N} \sup_{t \in T} |\inr{a_i,t}|^2 \|_{\psi_1} \leq c_4 \left(\ell^2(T)+d^2(T) \log N\right).
$$
\end{Theorem}

\subsection{The Recovery Algorithm}
\label{sec:alg}
The assumptions we make throughout this section are as follows:
\begin{Assumption}
Assume that $a$ is isotropic and subgaussian, and that
the noise $w$ in (\ref{eq:measn}) is a symmetric, $\psi_2$ random variable that is independent of $a$.
\end{Assumption}
\label{asp:noise}

Recall that the goal is to find an estimate $\hx$ of $x_0$ that is close to $x_0$ or to $-x_0$. Given the measurements $(y_i)_{i=1}^N$, a reasonable approach is to seek a value of $x$ that minimizes the empirical risk function:
\begin{equation}
\label{eq:pnlx}
P_N \ell_x=\frac{1}{N}\sum_{i=1}^N \big | |\inr{a_i,x}|^2 - y_i \big|^p,
\end{equation}
 for some $p$. Here we will consider values of $p$ in the regime $1<p \leq 2$; the exact choice of $p$ will become clear later on.
Note that for every $x \in T$,
\begin{equation}
\label{eq:lx}
\ell_x=\big | |\inr{a,x}|^2 - y \big|^p=| \inr{a,x-x_0}\inr{a,x+x_0}-w|^p.
\end{equation}
Since the empirical average $P_N\ell_x$ is not a convex function in $x$, it is impossible in general to find a value of $x$ that minimizes it. Luckily, for our purposes, one does not need an exact minimizer. Instead, in order to bound the estimation error, it is sufficient to find a value of $x$ for which the empirical risk is bounded above, as incorporated in the Definition \ref{def:hax-t} below.
To this end, one may use any algorithm for phase minimization and check whether the resulting solution satisfies the bound. Particularly useful in this context are techniques that depend on the initial starting point; such methods can be started from several different points, and in that way, if a particular solution does not satisfy the bound then the algorithm may be used again, but from a different starting point. Eventually, with high probability, a point satisfying the bound will be obtained. One algorithm of this form is the GESPAR method developed in \cite{SBE12}.
\begin{Definition} \label{def:hax-t}
Let $1<p \leq 2$ be given, and choose a value of $u \geq 1$. Given the data $(a_i,y_i)_{i=1}^N$, $\hx \in T$ is called a good estimate if it satisfies that
\begin{equation} \label{eq:hax-t}
\frac{1}{N}\sum_{i=1}^N \left| |\inr{a_i,\hx}|^2 - y_i \right|^p \leq \E|w|^p + u\left(Q_{T,N,W}-\frac{\||w|^p\|_{\psi_1}}{\sqrt{N}}\right),
\end{equation}
where
\begin{equation} \label{eq:qtn}
Q_{T,N}= d(T) \frac{\ell(T)}{\sqrt{N}} + \frac{\ell^2(T)}{N}, \ \ \ Q_{T,N,W}=Q_{T,N} + \frac{\||w|^p\|_{\psi_1}}{\sqrt{N}},
\end{equation}
and $\ell(T),d(T)$ are defined by \eqref{eq:lt},\eqref{eq:dt}.
\end{Definition}

To motivate the choice of $\hx$ in Definition~\ref{def:hax-t}, observe that
$Q_{T,N,W}$ captures the ``statistical complexity" of the problem -- namely, the sum of the ``gaussian complexity" of $T$, $Q_{T,N}$, and the influence of the noise, $\||w|^p\|_{\psi_1}/\sqrt{N}$. The parameter $u$ tunes the probability estimate, for the moment is of secondary importance. The exact choice of $p$ and $u$ will be specified in Theorem \ref{thm:noisy-main}.

This approach is based on a modified empirical risk minimization -- modified in two ways. First, instead of minimizing the loss functional $P_N \ell_x$, the search is for an empirical feasible point; some $x \in T$ for which
\begin{equation} \label{eq:hax-x-emp}
\frac{1}{N}\sum_{i=1}^N \left| |\inr{a_i,x}|^2 - y_i \right|^p - \frac{1}{N}\sum_{i=1}^N |w_i|^p \leq uQ_{T,N,W}.
\end{equation}
Observe that the value on the left hand side of \eqref{eq:hax-x-emp} is the empirical excess risk $P_N {\cal L}_x$ where
\begin{equation}
\label{eq:Lx}
{\cal L}_x = \ell_x -\ell_{x_0} = | \inr{a,x-x_0}\inr{a,x+x_0}-w |^p -
 |w|^p,
\end{equation}
is the excess loss functional. The definition implies that the empirical excess risk at $\hx$ is of the same order of magnitude as the ``statistical error" and thus
$$
P_N {\cal L}_{\hx} \leq uQ_{T,N,W}.
$$
One of the key components of the proof is to show that if $P_N {\cal L}_{\hx}$ is small, then so is the conditional expectation, $\E {\cal L}_{\hx}$. The second key component is that if $\E {\cal L}_{\hx}$ is small, then so is $\|\hat{x}-x_0\|_2 \|\hat{x}+x_0\|_2$.

Unfortunately, it is impossible to estimate the empirical excess risk since one does not have access to the sampled noise $w_1,...,w_N$, and therefore, nor to $\frac{1}{N}\sum_{i=1}^N |w_i|^p$ -- which is the reason for the second modification. By Assumption~\ref{asp:noise}, $w  \in L_{\psi_2}$ and consequently $|w|^p \in L_{\psi_1}$. From Corollary~\ref{cor:ewp},  if $u \leq N$, then with probability at least $1-2\exp(-c_1u^2)$,
$$
\frac{1}{N} \sum_{i=1}^N |w_i|^p \geq \E|w|^p - u\frac{\||w|^p\|_{\psi_1}}{\sqrt{N}}.
$$
Therefore, if $\hx$ satisfies \eqref{eq:hax-t}, then it also satisfies \eqref{eq:hax-x-emp}, meaning that its empirical excess risk is bounded above by the desired quantity.
This leads to the following proposition.
\begin{Proposition}
\label{prop:pnlx}
There exists an absolute constant $c_1$ for which the following holds. Let $\hx$ be a  point that satisfies \eqref{eq:hax-t} and let $w  \in L_{\psi_2}$. If $0 \leq u \leq N$, then with probability at least $1-2\exp(-c_1u^2)$, $P_N {\cal L}_{\hx} \leq uQ_{T,N,W}$.
\end{Proposition}

To see that there is always a point $\hat{x}$ that satisfies \eqref{eq:hax-t}, observe that for $x_0$ and $0 < u \leq N$, with probability at least $1-2\exp(-cu^2)$ (see Corollary~\ref{cor:ewp}),
\begin{align*}
\frac{1}{N}\sum_{i=1}^N \left| |\inr{a_i,x_0}|^2 - y_i\right|^p =& \frac{1}{N} \sum_{i=1}^N |w_i|^p \leq \E |w|^p +u \frac{\| |w|^p\|_{\psi_1}}{\sqrt{N}}
\\
\leq & \E |w|^p + u Q_{T,N,W}.
\end{align*}
Moreover, unless $T$ is very small and $W$ is very large, $Q_{T,N}$ is the dominant term in $Q_{T,N,W}$. For example, consider the case in which $w$ is a centered Gaussian with variance $\sigma$ and $T$ is the set of $k$-sparse vectors on the unit sphere. Then, $\||w|^p\|_{\psi_1}/{\sqrt{N}} \sim \sigma^p/\sqrt{N}$, while $Q_{T,N} \sim \sqrt{k\log(en/k)}/\sqrt{N}$ which  clearly is larger than $\||w|^p\|_{\psi_1}/{\sqrt{N}}$, as long as $k$ is large relative to $\sigma$.

We are now ready to state our main result. To this end recall the definition of $\kappa(s,t)$ given by \eqref{eq:kappa}, and let $\kappa_T=\inf_{s,t\in T} \kappa(s,t)$.
\begin{Theorem} \label{thm:noisy-main}
For every $\kappa>0$ and every $L \geq 1$ there exists constants $c_1,c_2,c_3$ and $c_4$ that depend only on $L$ and $\kappa$, for which the following holds. Let $a$ be distributed according to an isotropic, $L$-subgaussian measure, and assume that $\kappa_T \geq \kappa$. Assume further that $\|w\|_{\psi_2} < \infty$.
For every integer $N$ set
$$
\beta_N = \max\{c_1\left((\|w\|_{\psi_1} + d^2(T)) \log N + \ell^2(T)\right),e\}
$$
and
$$p=1+1/\log \beta_N.$$ Let $\hx$ be chosen to satisfy \eqref{eq:hax-t}.
Then, for $u \geq c_2$, with probability at least $1-2\exp(-c_3u^{1/3})$,
$$
\|x_0-\hx\|_2  \|x_0+\hx\|_2 \leq c_4 (uQ_{T,N,W})^{1/p} \sqrt{\log \beta_N},
$$
where $Q_{T,N,W}$ is defined by (\ref{eq:qtn}).
\end{Theorem}
Note that $\|w\|_{\psi_2} < \infty$ implies that $\|w\|_{\psi_p} < \infty$ for any $p \leq 2$.

Since $Q_{T,N,W}$  decays as $\sqrt{N}$ while $\beta_N$ grows as $\log N$, it is always possible to choose $N$ large enough so that the error given in the theorem is made sufficiently small.
As an example, consider the case of $k$-sparse vectors on the sphere. Hence, $d(T)=1$, and recall from Section~\ref{sec:examples} that $\ell(U_k) \sim (k\log(en/k))^{1/2}$. If $w$ is $L$-subgaussian then
$\|w\|_{\psi_1} \lesssim_L \sigma$ where $\sigma$ is the noise standard deviation, and
$$
\beta_N \sim_L (\sigma+1)\log N + k\log(en/k).
$$
If $k\log(en/k) \geq (\sigma+1) \log N$ (which is the reasonable range, as one expects $N  \sim k$ up to logarithmic factors), then $\beta_N \sim k\log(en/k)$,
and by Theorem~\ref{thm:noisy-main},
\begin{equation}
\label{eq:noisy-sparse}
 \|\hx-x_0\|_{2}  \|\hx+x_0\|_2  \lesssim_{L,\kappa}  \sqrt{\log \beta_N}  \left(\frac{k\log(en/k)}{N}\right)^{1/2-c/\log \beta_N},
\end{equation}
where $c$ is a constant.

To proceed, and as will be noted in Section \ref{sec:linear}, in the case of linear measurements, with high probability,
\begin{equation*}
\|\hx-x_0\|_{2}^2  \lesssim_{L}   \left(\frac{k\log(en/k)}{N}\right)^{1/2}.
\end{equation*}
To compare the ``quadratic" estimate with the linear one, note that if $N \leq (k\log(en/k))^\gamma$ for $\gamma \geq c_1$ and some constant $c_1 \geq 1$, then recalling that for every $x$, $x^{1/\log x} \leq e$, it is evident that
$$
\left( \frac{k \log(en/k)}{N}\right)^{-c/\log \beta_N}=\left(
k \log(en/k)\right)^{c(\gamma-1)/\log \beta_N} \sim \beta_N^{c(\gamma-1)/\log \beta_N}\sim C^{\gamma-1},
$$
where $C$ is an absolute constant. Therefore, with this choice of $N$,
$$
\|\hx-x_0\|_{2}  \|\hx+x_0\|_2  \lesssim_{L,\kappa,\gamma} \sqrt{\log (k\log(en/k))}  \left( \frac{k \log(en/k)}{N}\right)^{1/2},
$$
and up to logarithmic factors scales as the estimate in the linear case.

Clearly, it suffices to take $N \gtrsim_{L,\gamma,\eps} k\log(en/k) \log k$ to ensure that $\|\hx-x_0\|_2  \|\hx+x_0\|_2 \leq \eps$, which is off only by a $\log k$ factor from the optimal estimate in the linear case.

\begin{Corollary} \label{cor-noisy-sparse}
For every $L \geq 1$ and $\kappa>0$ there exist constants $c_1$, $c_2,c_3$ that depend only on $L$ and $\kappa$ and for which the following holds. Let $T$ be the set of $k$-sparse vectors on the sphere, set $a$ to be distributed according to an isotropic, $L$-subgaussian measure and assume that $\kappa_T \geq \kappa$. If the noise $w$ is $L$-subgaussian, $N \leq (k\log(en/k))^\gamma$ for $\gamma \geq c_1 \geq 1$ and $u>c_2$, then  with probability at least $1-2\exp(-c_3u^{1/3})$,
$$
\|\hx-x_0\|_{2}  \|\hx+x_0\|_2  \lesssim_{L,\kappa,\gamma,u} \sqrt{\log (k\log(en/k))}  \left( \frac{k \log(en/k)}{N}\right)^{1/2}.
$$
In particular, if $N \gtrsim_{L,\gamma,\eps,\delta} k\log(en/k) \log k$ then $\|\hx-x_0\|_2  \|\hx+x_0\|_2 \leq \eps$ with probability at least $1-\delta$.
\end{Corollary}

\subsection{Proof of Theorem~\ref{thm:noisy-main}}
\label{sec:proof}
The proof of the theorem requires several preliminary facts about empirical and Bernoulli processes. We refer the reader to \cite{LT,KW} for more details on these processes.

Throughout this section, $(\Omega,\mu)$ is a probability space and $(X_i)_{i=1}^N$ are iid, distributed according to $\mu$.  Let $\eps_1,...,\eps_N$ be independent, symmetric, $\{-1,1\}$-valued random variables, that are independent of $X_1,...,X_N$.

The first result we require is the contraction inequality for Bernoulli processes.
\begin{Theorem} \label{thm:contraction-Bernoulli} \cite{LT}
Let $F:\R^+ \to \R^+$ be convex and increasing and let $\phi_i:\R \to \R$ satisfy that $\phi(0)=0$ and  $\max_{1 \leq i \leq N} \|\phi_i\|_{\rm lip} \leq A$, where $\|\phi_i\|_{\rm lip}$ is the Lipschitz constant of $\phi_i$. Then, for any bounded $T \subset \R^N$,
$$
\E F \left(\frac{1}{2A}\sup_{t \in T} \left|\sum_{i=1}^N \eps_i \phi_i(t_i) \right| \right) \leq  \E F \left(\sup_{t \in T} \left|\sum_{i=1}^N \eps_i t_i \right| \right).
$$
\end{Theorem}
The following symmetrization argument allows one to bound an empirical process using the Bernoulli process indexed by the random set $\{ (h(X_i))_{i=1}^N : h \in H\}$.
\begin{Theorem} \label{thm:symmetrization} \cite{VW}
Let $F:\R^+ \to \R^+$ be convex and increasing and let ${\cal H}$ be a class of functions. Then
\begin{align*}
\E F \left(\sup_{h \in H} \left|\frac{1}{N}\sum_{i=1}^N h(X_i) -\E h \right| \right) & \leq \E F \left(2 \sup_{h \in H} \left|\frac{1}{N}\sum_{i=1}^N \eps_i h(X_i) \right|\right)
\\
& \leq \E F \left(4 \sup_{h \in H} \left|\frac{1}{N}\sum_{i=1}^N h(X_i) -\E h \right| \right).
\end{align*}
\end{Theorem}
We will use Theorem \ref{thm:contraction-Bernoulli} and Theorem \ref{thm:symmetrization} with $F(x)=|x|^q$ for $q \geq 2$.

The final result we require is the Kahane-Khintchine inequality \cite{LT}, on the moments of Bernoulli processes.
\begin{Theorem} \label{thm:Kahane-Khintchine}
There exists an absolute constant $c$ for which the following holds. If $T \subset \R^N$ and $q \geq 2$ then
\begin{align*}
\left\|\sup_{x \in T} \left|\sum_{i=1}^N \eps_i t_i\right| \right\|_{L_q}  \leq  c \sqrt{q} \E \sup_{x \in T} \left|\sum_{i=1}^N \eps_i t_i \right|
=  c\sqrt{q}\left\|\sup_{x \in T} \left|\sum_{i=1}^N \eps_i t_i\right| \right\|_{L_1}.
\end{align*}
\end{Theorem}

The first step in the proof of Theorem~\ref{thm:noisy-main} is to obtain an oracle inequality for $\E {\cal L}_x$ that holds for any $x \in T$ (see Lemma~\ref{lemma:noisy-oracle-inequality} below). The oracle inequality is used for $x=\hx$, and noting that for a good $\hx$, $P_N {\cal L}_{\hx}$ is bounded above by $uQ_{T,N,W}$ leads to an upper bound of the form  $\E {\cal L}_{\hx} \lesssim uQ_{T,N,W}$.
The second part of the proof consists of establishing a lower bound on $\E {\cal L}_{\hx}$ which is a function of
$\|x_0-\hx\|_2^p  \|x_0+\hx\|_2^p$.

\begin{Lemma} \label{lemma:noisy-oracle-inequality}
For every $L \geq 1$ there exist constants $c_1,c_2$ and $c_3$ that depend only on $L$ for which the following holds. If $p$ is chosen as in Theorem \ref{thm:noisy-main}, then for $u \geq c_1$, with probability at least $1-2\exp(-c_2u^{1/3})$, for every $x \in T$,
$$
\E {\cal L}_x \leq P_N{\cal L}_x + c_3u \left(d(T) \frac{\ell(T)}{\sqrt{N}} + \frac{\ell^2(T)}{N}\right).
$$
\end{Lemma}
The choice made above, of $p=1+1/\log \beta_N$ is the key point in the proof, and it is there to balance two issues. On one hand, if $p$ is larger than $1$, then the empirical process $x \to P_N {\cal L}_x - \E {\cal L}_x$ becomes much harder to control. On the other, if $p=1$, then the loss is not strictly convex, and it is impossible to lower bound $\E {\cal L}_x $ using a function of $\|x-x_0\|_2 \|x+x_0\|_2$. This choice of $p$ is sufficiently close to $1$ to enable control of the empirical process (the main point is \eqref{eq:D-control}), while it is far enough from $1$ to give enough convexity to enable the lower bound.

\proof
Fix $q \geq 2$.
By the symmetrization theorem (Theorem \ref{thm:symmetrization}) and the independence of $a$ and $W$,
$$
\E \sup_{x \in T}  |P_N {\cal L}_x -
 \E {\cal L}_x|^q \leq  \E \E_\eps \sup_{x \in T} \left|\frac{2}{N}\sum_{i=1}^N \eps_i \left(|\inr{a_i,x-x_0} \inr{a_i,x+x_0}-w_i|^p -|w_i|^p\right) \right|^q.
$$
Let
$$
D_{\infty,N} = 2\max_{1 \leq i \leq N} \left(|w_i|+\sup_{x \in T} |\inr{a_i,x-x_0} \inr{a_i,x+x_0}|\right),
$$
and observe that for every realization of $(w_i)_{i=1}^N$, the functions $y \to |y-w_i|^p-|w_i|^p$ vanish at $0$ and are Lipschitz on $[-b,b]$ with a constant $p(b+|w_i|)^{p-1}$. For $b \leq \max_{1 \leq i \leq N} \sup_{x \in T}|\inr{a_i,x-x_0} \inr{a_i,x+x_0}|$ this constant is proportional to $D_{\infty,N}^{p-1}$, since $p \leq 2$. Applying the contraction inequality (Theorem \ref{thm:contraction-Bernoulli}), conditioned on $w_1,...,w_N$ and $a_1,...,a_N$,
\begin{align*}
& \E_\eps \sup_{x \in T} \left|\sum_{i=1}^N \eps_i \left(|\inr{a_i,x-x_0} \inr{a_i,x+x_0}-w_i|^p -|w_i|^p\right) \right|^q
\\
\leq &
(cD_{\infty,N}^{p-1})^q  \E_\eps \sup_{x \in T} \left|\sum_{i=1}^N \eps_i \inr{a_i,x-x_0} \inr{a_i,x+x_0} \right|^q.
\end{align*}

By the Kahane-Khintchine inequality, the Cauchy-Schwarz inequality, and Jensen's inequality combined with reverse symmetrization (the other direction of Theorem \ref{thm:symmetrization}),
\begin{align*}
& \E_{a \times W}  \left[ (D_{\infty,N}^{p-1})^q  \E_\eps \sup_{x \in T} \left|\sum_{i=1}^N \eps_i \inr{a_i,x-x_0} \inr{a_i,x+x_0} \right|^q \right]
\\
\leq & (c_1\sqrt{q})^q  \E_{a \times W} \left[(D_{\infty,N}^{p-1})^q
\left(\E_\eps \sup_{x \in T} \left|\sum_{i=1}^N \eps_i \inr{a_i,x-x_0} \inr{a_i,x+x_0} \right| \right)^q \right]
\\
\leq & (c_1\sqrt{q})^q \|D_{\infty,N}^{p-1}\|_{L_{2q}}^q  \left(\E_{a \times W} \left(\E_\eps \sup_{x \in T} \left|\sum_{i=1}^N \eps_i \inr{a_i,x-x_0} \inr{a_i,x+x_0} \right|\right)^{2q} \right)^{1/2}
\\
\leq & (c_2\sqrt{q})^q \|D_{\infty,N}^{p-1}\|_{L_{2q}}^q  \left\|\sup_{x \in T}\left| \frac{1}{N} \sum_{i=1}^N \inr{a_i,x-x_0} \inr{a_i,x+x_0}- \E \inr{a_i,x-x_0} \inr{a_i,x+x_0} \right|\right\|_{L_{2q}}^q,
\end{align*}
where $\| \ \|_{L_{2q}}$ is taken with respect to the $N$-product measure $(a \otimes w)^N$.

Setting
$$
B_{T,N,q}=\left\|\sup_{x \in T}\left| \frac{1}{N} \sum_{i=1}^N \inr{a_i,x-x_0} \inr{a_i,x+x_0}- \E \inr{a_i,x-x_0} \inr{a_i,x+x_0} \right|\right\|_{L_{2q}},
$$
it follows that
$$
\|\sup_{x \in T}|P_N {\cal L}_x - \E {\cal L}_x| \|_{L_q} \leq c_2 \sqrt{q} \|D_{\infty,N}^{p-1}\|_{L_{2q}} B_{T,N,q},
$$
and it remains to bound $\|D_{\infty,N}^{p-1}\|_{L_{2q}}$ and $B_{T,N,q}$.

To estimate $B_{T,N,q}$, set $\bar{T}=\{x-s : s,t \in T \cup (-T)\}$. Note that $\ell(\bar{T}) \leq 2\ell(T)$ and that $d(\bar{T}) \leq 2d(T)$. Applying Theorem~\ref{thm:main-emp-est},
$$
B_{T,N,q} \leq cq^{3/2}\left(d(T) \ell(T)/ \sqrt{N} + \ell^2(T)/N\right).
$$
Turning to $\|D_{\infty,N}^{p-1}\|_{L_{2q}}$, observe that pointwise
$$
\max_{1 \leq i \leq N} \sup_{x \in T} |\inr{a_i,x-x_0} \inr{a_i,x+x_0}| \leq  \max_{1 \leq i \leq N} \sup_{x \in T} |\inr{a_i,x}|^2.
$$
Therefore, by Theorem \ref{thm:monotone-psi-1},
$$
\|\max_{1 \leq i \leq N} \sup_{x \in T} |\inr{a,x-x_0} \inr{a,x+x_0}| \|_{\psi_1} \lesssim_L \ell^2(T) + d^2(T) \log N.
$$
Hence, by \eqref{eq:max-orlicz}
\begin{align*}
\|D_{\infty,N}\|_{\psi_1} \leq & 2\left(\|\max_{1 \leq i \leq N} w_i\|_{\psi_1} + \|\max_{1 \leq i \leq N} \sup_{x \in T} |\inr{a,x-x_0} \inr{a,x+x_0}| \|_{\psi_1}\right)
\\
\leq & c_L \left(\left(\|w\|_{\psi_1}+d^2(T) \right) \log N + \ell^2(T)\right) \equiv \beta_N.
\end{align*}
Set $p=1+1/\log \beta_N$. With this choice, combined with the moment characterization of the $\psi_1$ norm \eqref{eq:orlicz-moments}, it is evident that
\begin{equation} \label{eq:D-control}
\|D_{\infty,N}^{p-1}\|_{L_{2q}} \leq c_5q^{p-1}
\end{equation}
for a suitable absolute constant $c_5$. Indeed,
$$
\left(\E D_{\infty,N}^{(p-1)2q}\right)^{1/(p-1)2q} \lesssim (p-1)q\|D_{\infty,N}\|_{\psi_1},
$$
and thus,
$$
\|D_{\infty,N}^{p-1}\|_{L_{2q}} \lesssim ((p-1)q)^{p-1} \|D_{\infty,N}\|_{\psi_1}^{p-1} \lesssim ((p-1)q)^{p-1} \beta_N^{1/\log \beta_N} \leq c_5q^{p-1}.
$$

With these two estimates, it is evident that there exists a constant $c_6$ that depends only on $L$ for which, for every $q \geq 2$,
\begin{align*}
\|\sup_{x \in T} |P_N {\cal L}_x - \E {\cal L}_x| \|_{L_q} \leq & c_6q^{1/2+3/2+(p-1)} \left(d(T) \frac{\ell(T)}{\sqrt{N}} + \frac{\ell^2(T)}{N}\right)
\\
\leq c_6q^{3} \left(d(T) \frac{\ell(T)}{\sqrt{N}} + \frac{\ell^2(T)}{N}\right).
\end{align*}
With this $L_q$ estimate at hand, it is standard to show (see, e.g., \cite{LAMA} for a similar argument), that for $u \geq 1$, with probability at least $1-2\exp(-c_7u^{1/3})$,
$$
\sup_{x \in T} |P_N {\cal L}_x - \E {\cal L}_x| \leq c_8 u\left(d(T) \frac{\ell(T)}{\sqrt{N}} + \frac{\ell^2(T)}{N}\right),
$$
where $c_7$ and $c_8$ depend only on $L$.

Finally, in this case, for every $x \in T$
$$
 \E {\cal L}_x \leq P_N {\cal L}_x + c_8u\left(d(T)\frac{\ell(T)}{\sqrt{N}} + \frac{\ell^2(T)}{N}\right),
$$
as claimed.
\endproof

With Lemma \ref{lemma:noisy-oracle-inequality} in mind, the choice of $\hx$ becomes clearer. One would like to find any point in $T$ for which $P_N {\cal L}_x$ is, at most, of the same order of magnitude as the combined complexity term of the set $T$ and the noise
$$
Q_{T,N,W} = \left(d(T) \frac{\ell(T)}{\sqrt{N}} + \frac{\ell^2(T)}{N}\right)+\frac{\| |w|^p\|_{\psi_1}}{\sqrt{N}}.
$$
In this case, $\E {\cal L}_x$ can be bounded above by $Q_{T,N,W}$.
It follows from Proposition~\ref{prop:pnlx} and Lemma~\ref{lemma:noisy-oracle-inequality}  that with probability at least $1-2\exp(-cu^{1/3})-2\exp(-cu^2)$, both
$$
P_N {\cal L}_{\hx} \leq uQ_{T,N,W}
$$
and
$$
\E {\cal L}_{\hx} \leq P_N {\cal L}_{\hx} + c_1u Q_{T,N,W}.
$$
Therefore,
\begin{equation} \label{eq:mini-condition}
 \E {\cal L}_{\hx} \leq c_2u Q_{T,N}.
\end{equation}

To complete the proof of Theorem \ref{thm:noisy-main}, one has to bound $\E {\cal L}_{\hx}$ from below. The fact that $p>1$ gives ``enough convexity" to establish the desired lower bound.

Given $x \in T$ set $h_x=\inr{a,x-x_0} \inr{a,x+x_0}$ and recall that ${\cal L}_x=|h_x(a)-w|^p-|w|^p$. Since $w$ is a symmetric random variable, it is distributed as $\eps|w|$, where $\eps$ is a symmetric $\{-1,1\}$-valued random variable, independent of $w$ and of $a$. Therefore,
\begin{align} \label{eq:lower-expectation}
\E {\cal L}_x =& \E_{a \times W} \E_\eps (\left|h_x(a)-\eps|w|\right|^p-|w|^p) \nonumber
\\
=& \E_{a \times W} \left(\frac{1}{2}\left| |w|-h_x(a) \right|^p + \frac{1}{2}\left| |w|+h_x(a) \right|^p-|w|^p\right).
\end{align}
It is well known (see, e.g., \cite{Mat,BEL} ) that if $1 < p \leq 2$, then for every $c,d \in \R$,
\begin{equation} \label{eq:pointwise-1}
\frac{1}{2}\left(|c+d|^p+|c-d|^p\right) \geq \left(c^2+(p-1)d^2\right)^{p/2}.
\end{equation}
 Observe that the function $f(t)=(t^2+(p-1)d^2)^{p/2}-t^p$ is increasing for $t \geq 0$ and that $f(0)=(p-1)^{p/2}d^p$. Hence, for every $c,d$,
$$
(c^2+(p-1)d^2)^{p/2}-c^p \geq (p-1)^{p/2}d^p.
$$
Taking $c=|w|$ and $d=|h_x(a)|$,
$$
\E {\cal L}_x \geq (p-1)^{p/2}\E |h_x(a)|^p = (p-1)^{p/2} \E |\inr{a,x-x_0}  \inr{a,x+x_0}|^p.
$$

By the definition of $\kappa(s,t)$, for every $s,t$ and $p \geq 1$,
$$
\left(\E |\inr{a,t} \inr{a,s} |^p\right)^{1/p} \geq \E |\inr{a,t}  \inr{a,s} | \geq  \kappa_T \|t\|_2  \|s\|_2 \geq   \kappa \|t\|_2  \|s\|_2.
$$
Therefore,
$$
\E {\cal L}_x \geq  \kappa^p (p-1)^{p/2} \|x-x_0\|_{2}^p  \|x+x_0\|_{2}^p.
$$
Combining this lower bound with \eqref{eq:mini-condition}, and recalling that $p=1+1/\log \beta_N$ completes the proof of the theorem.

\subsection{Examples}

Let us present some of the examples seen in Section~\ref{sec:examples}, in the noisy setting. Other examples may be obtained with similar ease.

In order to apply the results of Theorem~\ref{thm:noisy-main}, one has to determine $d(T),\ell(T),\E|w|^p,\|w\|_{\psi_1},\|w\|_{\psi_2}$ and $\||w|^p\|_{\psi_1}$.

In all the examples below we will assume that $T \subset S^{n-1}$ and so $d(T)=1$. Since $w$ is symmetric and $L$-subgaussian, then $\|w\|_{\psi_1} \leq \|w\|_{\psi_2} \lesssim L\sigma$, where $\sigma$ is the noise variance. Also, since $1 < p \leq 2$, $\| |w|^p\|_{\psi_1}=\|w\|_{\psi_p}^p \lesssim (L\sigma)^p$.

\subsubsection{Entire Space $T=\R^n$}
If $T=S^{n-1}$ then $\ell(T) \sim \sqrt{n}$,
implying that
$$
Q_{T,N} \lesssim \left(\sqrt{\frac{n}{N}}+\frac{n}{N}\right) \lesssim \sqrt{\frac{n}{N}}
$$
for the regime of $N$ we are interested in, and
$$
Q_{T,N,W} \lesssim \sqrt{\frac{n}{N}}+\frac{\sigma^p}{\sqrt{N}}.
$$
In addition,
$$
\beta_N \sim (\sigma+1)\log N+n.
$$
Suppose that $n \geq (\sigma+1)\log N$. Then, by Theorem~\ref{thm:noisy-main},
\begin{equation}
\label{eq:noisy-whole-space}
 \|\hx-x_0\|_{2}  \|\hx+x_0\|_2  \lesssim_{L,\kappa}  \sqrt{\log \beta_N}  \left(\frac{n}{N}\right)^{1/2-c/\log \beta_N},
\end{equation}
where $c$ is an absolute constant.
 If $N \leq n^\gamma$ for $\gamma \geq c_1 \geq 1$, then
$$
\left(\frac{n}{N}\right)^{-c/\log \beta_N}=n^{c(\gamma-1)/\log \beta_N} \sim \beta_N^{c(\gamma-1)/\log \beta_N}\sim C^{\gamma-1},
$$
for a suitable absolute constant $C$. Therefore, with this choice of $N$,
$$
\|\hx-x_0\|_{2}  \|\hx+x_0\|_2  \lesssim_{L,\kappa,\gamma,u} \sqrt{\log n}  \left( \frac{n}{N}\right)^{1/2},
$$
and it suffices to take $N \gtrsim_{L,\gamma,\kappa,\eps,\delta} n \log n$ to ensure that $\|\hx-x_0\|_2  \|\hx+x_0\|_2 \leq \eps$ with probability at least $1-\delta$.

\begin{Corollary} \label{cor-noisy-whole-space}
For every $L \geq 1$ and $\kappa>0$ there exist constants $c_1$, $c_2,c_3$ that depend only on $L$ and $\kappa$ and for which the following holds. If $\mu$, $a$ and $w$ are as above, $T=S^{n-1}$ and $N \leq n^\gamma$ for $\gamma \geq c_1 \geq 1$, then for $u>c_2$  with probability at least $1-2\exp(-c_3u^{1/3})$,
$$
\|\hx-x_0\|_{2}  \|\hx+x_0\|_2  \lesssim_{L,\kappa,\gamma,u} \sqrt{\log n}  \left( \frac{n}{N}\right)^{1/2}.
$$
\end{Corollary}

\subsubsection{Sparse Vectors}
We already treated the case of sparse vectors in
Corollary~\ref{cor-noisy-sparse}. The block-sparse setting can be treated in a similar manner, leading to the following corollary.
\begin{Corollary} \label{cor-noisy-sparse-block}
For every $L \geq 1$ and $\kappa>0$ there exist constants $c_1$, $c_2,c_3$ that depend only on $L$ and $\kappa$ and for which the following holds. If $a$ and $w$ are as above, $T$ is the set of $k$-block sparse vectors of length $d$ on the sphere and $N \leq (k\log(en/dk)+dk)^\gamma$ for $\gamma \geq c_1 \geq 1$, then for $u>c_2$  with probability at least $1-2\exp(-c_3u^{1/3})$,
$$
\|\hx-x_0\|_{2}  \|\hx+x_0\|_2  \lesssim_{L,\kappa,\gamma,u} \sqrt{\log (k\log(en/dk)+dk)}  \left( \frac{k \log(en/dk)+dk}{N}\right)^{1/2}.
$$
In particular, if $N \gtrsim_{L,\gamma,\eps,\delta} (k\log(en/kd) +dk)(\log k+\log d)$ then $\|\hx-x_0\|_2  \|\hx+x_0\|_2 \leq \eps$ with probability at least $1-\delta$.
\end{Corollary}

\section{Connection with Results on Linear Estimation}
\label{sec:linear}

It should come as no surprise that the methods used here are very similar in nature to the analogous ``linear questions". Both stability and noisy recovery are well understood in the linear case, and in a sharp way, as we will explain below.

First, consider the question of stability. Suppose that the measurements are given by $y=Ax$ for some $N \times n$ matrix $A$. In the linear setting, a natural notion of stability in a set $T \subset \R^n$ is that for all $s,t \in T$,
\begin{equation}
\|A t-A s\|_{2} \geq C \|t-s\|_{2},
\end{equation}
where $C$ is a constant.

Note that here the $\ell_2$ norm is used in the left hand side, rather than the $\ell_1$ norm. An $\ell_2$ stability result is superior to an $\ell_1$ estimate, simply because the $\ell_2$ norm is smaller. And, It is natural to compare an $\ell_2$ stability result in the linear case to the $\ell_1$ stability result for quadratic measurements we established.

Stability in a set $T$ for a random ensemble depends on the way in which a typical operator acts on the set
$$
T_-=\{\frac{t-s}{\|t-s\|_2} : t \not = s, \ t \in T\} \subset S^{n-1}.
$$
Indeed, because $a$ is distributed according to an isotropic measure, for every $z \in S^{n-1}$, $\E |\inr{a,z}|^2 =1$. Thus, stability on $T$ is equivalent to an estimate on
\begin{equation} \label{eq:quad-proc}
\sup_{z \in T_-} \left|\frac{1}{N}\sum_{i=1}^N |\inr{a_i,z}|^2-1 \right|,
\end{equation}
which is strictly smaller than $1$.

With this in mind, the stability constant in $\R^n$ is a lower bound on the smallest singular value of a typical operator from the given random ensemble.

The study of the process \eqref{eq:quad-proc}, both for $T_-=S^{n-1}$ and for an arbitrary subset of the sphere has been extensive in recent years.
A good starting point for the interested reader would be  \cite{KM,MPT} for subgaussian ensembles, \cite{ALPT,Men-GAFA} for log-concave ensembles, and \cite{SV,MenPao1,MenPao2} for ensembles with heavy tails (though this does not begin to cover the extensive literature on the topic).

In the context of this paper, subgaussian ensembles, the best estimate on \eqref{eq:quad-proc} follows from Theorem \ref{thm:main-emp-est}, applied to the class $F=H=\{\inr{v,\cdot}, \ v \in T_-\}$. Moreover, in \cite{Men-orc} it was shown that under very mild assumptions on the set $T_-$, the estimate is sharp.
\begin{Theorem} \label{thm:linear-stabilty}
For every $L \geq 1$ there exist constants $c_1,c_2$ and $c_3$ that depend only on $L$ for which the following holds. If $T \subset \R^n$ and $a$ is distributed according to an isotropic, $L$-subgaussian measure, then
for $u \geq c_1$, with probability at least $1-2\exp(-c_2u\ell(T_-))$, for every $s,t \in T$,
$$
\| As-At \|_{2} \geq  \|s-t\|_{2}/\sqrt{2},
$$
provided that $N \geq c_3 u^3\ell^2(T_-)$.
\end{Theorem}
\proof
Since
$$
\frac{\|A t-A s\|_{2}^2}{\|t-s\|_2^2} = \sum_{i=1}^N  |\inr{a_i,z}|^2,
$$
for $z=(t-s)/\|t-s\|_2 \in T_-$, then setting $z_{s,t}=N^{-1}\sum_{i=1}^N  |\inr{a_i,z}|^2$, it suffices to bound $\inf_{s,t \in T, \ s \not = t} z_{s,t}$ from below. Since $a$ is isotropic, $\E z_{s,t}=\E |\inr{a, z}|^2=1$. Applying Theorem \ref{thm:main-emp-est} for $N \gtrsim_L u^6 \ell^2(T_-)$ and recalling that $T_-\subset S^{n-1}$, it follows that with probability at least $1-2\exp(-cu^2 \ell(T_-))$,
$$
\sup_{z \in T_-} \left|\frac{1}{N} \sum_{i=1}^N \inr{z,a}^2 -1 \right| \lesssim_L u^3 \left(\frac{\ell(T_-)}{\sqrt{N}} + \frac{\ell^2(T_-)}{N}\right) \leq 1/2.
$$
On that event, for every $s \not = t$,
$$
\|As -At \|_2^2 = \|A(s-t)\|_2^2 \geq \|s-t\|_2^2/2,
$$
as claimed.
\endproof

Observe that the same complexity parameter appears in the linear case as in the ``quadratic" stability result -- the gaussian complexity of a ``projection" of $T-T$ onto the sphere (and, of course, the $T+T$ component does not appear). In all the examples we presented in this note, $T+T$ and $T-T$ have essentially (or exactly) the same complexity, and thus the stability estimates in the linear case coincide with quadratic bounds, as will be the case for any $T \subset \R^n$ with a similar property. Thus, in these cases, there is no harm in requiring stability over quadratic measurements rather than with respect to linear ones.

The noisy recovery problem in the linear case is much simpler, since the resulting empirical process is well behaved even if one uses the squared loss functional. The advantage in considering the squared loss functional is that one has the benefit of the required convexity ``for free". With this objective, noisy recovery becomes a linear regression problem in $\R^n$, indexed by $T$. This is a well studied topic in statistics. We refer the reader to \cite{Kolt} for relatively recent results related to this question.

The best results to-date on linear regression that take into account the complexity of the indexing set $T$ can be found in \cite{Men-orc}. One may show that these estimates are sharp under very mild assumptions on $T$, and it turns out that these assumptions are satisfied in the examples that were presented here. Since our bounds in the ``quadratic" case are of the same order of magnitude as in the easier, linear case, and since these bounds are optimal in the linear case, it is reasonable to expect that they are optimal in the quadratic scenario as well. Unfortunately, the methods required to prove this optimality are rather involved, and we will not explore this issue here. Rather we refer the reader to \cite{Men-orc}, in which the linear case is explored.

\bibliographystyle{plain}
\bibliography{references}

\end{document}